\documentclass[twoside,11pt,reqno]{amsart}

\usepackage[T1]{fontenc}
\usepackage{txfonts}

\DeclareSymbolFont{symbols}{OMS}{cmsy}{m}{n}

\usepackage[margin=1.2in]{geometry}

\newcommand{\myrtimes}{\rtimes}

\usepackage[utf8]{inputenc}
\usepackage[T1]{fontenc}

\usepackage{mathrsfs}
\usepackage{mathtools}
\usepackage{amsfonts}
\usepackage{amsmath,amssymb,amsthm}

\usepackage{enumerate}
\usepackage{graphicx,tikz}
\usepackage{ifthen}
\usetikzlibrary{arrows,scopes}
\usetikzlibrary{calc}

\usepackage{setspace}

\usepackage{hyperref}
\usepackage[alphabetic,initials]{amsrefs}

\newtheorem{thm}{Theorem}[section]

\newtheorem{cor}[thm]{Corollary}
\newtheorem{prop}[thm]{Proposition}
\theoremstyle{definition}
\newtheorem{defi}[thm]{Definition}
\newtheorem{conjecture}[thm]{Conjecture}
\newtheorem{question}[thm]{Question}
\newtheorem{example}[thm]{Example}
\theoremstyle{remark}
\newtheorem{rmk}[thm]{Remark}

\newcommand{\Hil}{\mathcal{H}}

\newcommand{\Spin}{\mathop{\mathsf{Spin}}}
\newcommand{\SO}{\mathop{\mathsf{SO}}}

\newcommand{\SU}{\mathop{\mathsf{SU}}}
\newcommand{\Gtwo}{\mathop{\mathsf{G}_2}}
\newcommand{\Ffour}{\mathop{\mathsf{F}_4}}
\newcommand{\grE}{\mathop{\mathsf{E}}}

\newcommand{\Mob}{\mathsf{M\ddot ob}}

\newcommand{\Sc}[1][]{\mathbb{S}^{1#1}}

\newcommand{\C}{\mathcal{C}}

\newcommand{\cF}{\mathcal{F}}

\newcommand{\cB}{\mathcal{B}}

\renewcommand{\k}{\mathcal{k}}

\newcommand{\cD}{\mathcal{D}}

\newcommand{\cC}{\mathcal{C}}

\newcommand{\A}{\mathcal{A}}
\newcommand{\cI}{\mathcal{I}}

\newcommand{\M}{\mathcal{M}}

\newcommand{\N}{\mathcal{N}}
\newcommand{\cG}{\mathcal{G}}

\newcommand{\CC}{\mathbb{C}}

\newcommand{\ZZ}{\mathbb{Z}}
\newcommand{\NN}{\mathbb{N}}

\DeclareMathOperator{\End}{End}
\DeclareMathOperator{\DHR}{DHR}
\DeclareMathOperator{\Rep}{Rep}

\DeclareMathOperator{\Hom}{Hom}

\DeclareMathOperator{\Mor}{Mor}
\DeclareMathOperator{\Vir}{Vir}

\DeclareMathOperator{\Ad}{Ad}

\DeclareMathOperator{\id}{id}

\DeclareMathOperator{\B}{B}

\DeclareMathOperator*{\Dim}{Dim}
\DeclareMathOperator*{\Mod}{Mod}
\DeclareMathOperator*{\Hilb}{Hilb}
\DeclareMathOperator{\Irr}{Irr}

\newcommand{\e}{\mathrm{e}}

\newcommand{\ima}{\mathrm{i}}

\newcommand{\op}{\mathrm{op}}

\usepackage{xspace}
\DeclareRobustCommand{\eg}{e.g.\@\xspace}
\DeclareRobustCommand{\cf}{cf.\@\xspace}
\DeclareRobustCommand{\ie}{i.e.\@\xspace}
\makeatletter
\DeclareRobustCommand{\etc}{%
    \@ifnextchar{.}%
        {etc}%
        {etc.\@\xspace}%
}
\makeatother

\newcommand{\Cstar}{$C^\ast$\@\xspace}

\newcommand{\bim}[4][]{{}\prescript{\vphantom{#1}}{#2}{#3}^{#1}_{#4}}

\newcommand{\LR}{\mathrm{LR}}
\renewcommand{\N}{N}
\renewcommand{\M}{M}
\newcommand{\ADEe}{$A$-$D_{2n}$-$E_{6,8}$\@\xspace}
\newcommand{\rev}[1]{{#1}^\mathrm{rev}}

\begin{document}
\date{\today}
\dateposted{\today}
\newcommand{\mytitle}{%
A Remark on 
CFT Realization of Quantum Doubles of Subfactors. Case Index $<4$%
}
\title{\mytitle}
\curraddr{Vanderbilt University, Department of Mathematics, 1326 Stevenson
Center, Nashville, TN 37240, USA}
\author{Marcel Bischoff}
\email{marcel.bischoff@vanderbilt.edu}
\thanks{Supported by NSF Grant DMS-1362138}
\begin{abstract}
It is well-known that the quantum double  $D(N\subset M)$  of a finite depth subfactor $N\subset M$, or equivalently the Drinfeld center of the even part fusion category, is a unitary modular tensor category. Thus, it should arise in conformal field theory.
We show that for every subfactor $N\subset M$ with index $[M:N]<4$ the quantum double $D(N\subset M)$ is realized as the representation category of a completely rational conformal net.
In particular, the quantum double of $E_6$ can be realized as a $\ZZ_2$--simple current 
extension of $\SU(2)_{10}\times \Spin(11)_1$ and thus is not exotic in any sense. 
As a byproduct, we obtain a vertex operator algebra for every such subfactor.

We obtain the result by showing that if a subfactor $N\subset M $ arises from $\alpha$--induction 
of completely rational nets $\A\subset \cB$ and there is a net $\tilde \A$ 
with the opposite braiding, then the quantum $D(N\subset M)$ is realized by 
completely rational net. We construct completely rational nets with the opposite braiding of  $\SU(2)_k$
and use the well-known fact that all subfactors with index $[M:N]<4$ arise by
$\alpha$-induction from  $\SU(2)_k$.
\end{abstract}

\maketitle

\setcounter{tocdepth}{3}
\setcounter{tocdepth}{1}
\tableofcontents

\newcommand{\bs}{$\backslash$}
\renewcommand{\L}{\mathrm{L}}

\newcommand{\tikzmath}[2][0.50]
{\vcenter{\hbox{\begin{tikzpicture}[scale=#1] #2\end{tikzpicture}}}
}
\newcommand{\colM}{black!20}
\newcommand{\colMa}{orange!50}
\newcommand{\colMab}{green!50}
\newcommand{\colMb}{red!30}
\renewcommand{\colMa}{black!20} \renewcommand{\colMab}{black!35} \renewcommand{\colMb}{black!30}
\newcommand{\colN}{black!10}
\newcommand{\mydot}[1]{\begin{scope}[shift={#1}] \fill[shift only] (0,0) circle (1.5pt); \end{scope}}

\section{Introduction}
A unitary fusion category can be seen as the generalization of 
a finite group $G$, which is neither assumed to be commutative nor co-commutative. In particular, the easiest examples are 
the category $\Rep(G)$ of unitary representation
of a finite group $G$ and the category $\Hilb_G$ of $G$-graded finite dimensional Hilbert spaces.
Note that $G$ is co-commutative in the sense that  
$\Rep(G)$ is commutative, while $G$ is in general non-commutative.

A factor is a von Neumann algebra with trivial center and a rather boring object. On the other hand a
subfactor, an inclusion $N\subset M$ of a factor $N$ into another, turns often out to be a really interesting object.
For example subfactor obtained by taking a fixed point with respect to a free action of a finite group
$N=M^G\subset M$ gives $\bim M\cF M\cong\Hilb_G$ and $\bim N \cF N\cong \Rep(G)$.
In general, a finite depth subfactor $N\subset M$ gives two unitary fusion categories $\bim N\cF N$ and $\bim M\cF M$ 
which are (higher) Morita equivalent.

Conversely, having a unitary fusion category $\cF$, there is a subfactor $N\subset M$, such that 
$\bim N\cF N \cong \bim M \cF M \cong \cF$. 
An important invariant \cite{Jo1983} is the index $[M:N]$ 
of a subfactor, which by Jones' index theorem takes values in:
$$
  [M:N]\in \left\{ 4\cos^2\left(\frac{\pi}{m}\right) : m=3,4,\ldots\right\}\cup [4:\infty].
$$
Another invariant is a pair of graphs, the principle and dual principal graphs, which are bipartite graphs. For 
index $[M:N]<4$ they are given by \ADEe Dynkin diagrams, 
where the index is related to the Coxeter number $m$ of the graph by $[M:N]=4\cos^2(\pi/m)$.

A unitary braided fusion category 
is a unitary fusion category with a braiding A braiding is a natural family of 
unitaries $\varepsilon(\rho,\sigma)\in\Hom(\rho\otimes \sigma,\sigma \otimes \rho)$.
Braided categories give a representation of the $n$-strand braid groups $B_n=\langle
e_1,\ldots,e_{n-1} : e_{i+1}e_ie_{i+1}=e_ie_{i+1}e_i,~e_ie_j=e_je_i \text{ if }|i-j|\geq2\rangle$ on $\Hom(\rho^{\otimes n},\rho^{\otimes n})$.
If $\varepsilon(\rho,\sigma)\varepsilon(\sigma,\rho)=1_{\sigma\otimes\rho}$ for all objects $\sigma,\rho$, it is called a symmetric fusion category. In this case the representations of the braid group are actually representations of the symmetric group.
On the other hand, in a unitary modular tensor category (UMTC)
the braiding is non-degenerated, in the sense that if $\varepsilon(\rho,\sigma)\varepsilon(\sigma,\rho)=1_{\sigma\otimes\rho}$
for all $\rho$, then $\sigma$ is a direct sum of the trivial object. 

Simple examples of UMTCs $\cC$ are the one where every irreducible object 
is invertible (has dimension 1). Then the fusion rules form an abelian group $A$
and $\cC$ is characterized by a non-degenerated quadratic form on $A$. 

The Drinfeld center of a UFC $\cF$, or the quantum double of a finite depth subfactor $N\subset M$, which equals the Drinfeld center $Z(\cF)$ of either of its fusion categories
$\cF\in\{\bim N\cF N, \bim M \cF M\}$, is a unitary modular tensor category \cite{Mg2003II}.

A coordinate version of modular tensor categories were invented by Moore and Seiberg \cite{MoSe1990} to axiomatize (the topological behaviour of) conformal field theories. Braided tensor categories also appeared in 
algebraic quantum field theory \cite{FrReSc1989} and UMTCs and their structure were 
analyzed by Rehren in \cite{Re1989}.
There are two axiomatizations for chiral CFT: vertex operator algebras (VOAs) and conformal nets and in both approaches the representation theory gives under certain sufficient conditions a (unitary) modular tensor category. 

The natural question arises, if all modular tensor categories arise as representation categories 
of chiral CFT.
A subquestion is if the quantum double of subfactors or equivalently Drinfeld centers of unitary fusion categories arise in this way.

We want to discuss such a question in the framework of conformal nets, which is naturally related to the study of subfactors. 
More precisely, if $\A$ is a completely rational conformal net, then the category of Doplicher--Haag--Roberts representions $\Rep(\A)$ 
is a unitary modular tensor category (UMTC) by \cite{KaLoMg2001}. 

We vaguely conjecture that the following is true.
\begin{conjecture} 
  \label{conj:1}
  Let $\cC$ be a unitary modular tensor category (UMTC), then there is at completely rational conformal net 
  $\A$ with $\Rep(\A)\cong\cC$.
\end{conjecture}
An analogous statement 
in higher dimensional algebraic quantum field theory (see \cite{Ha}) is known to be true.
Namely, it is shown that under natural assumption
for every net $\A$ there is a compact (metrizable) group $G$ with a central involutive element $k\in G$,
such that the category of Doplicher--Haag--Roberts representations $\DHR(\A)$ is the category of unitary representations of $G$, which is $\ZZ_2$-graded by $k$.
Every such pair  $\{G,k\}$ can be realized 
using free field theory \cite{DoPi2002}.

Conjecture \ref{conj:1}
 would imply the following weaker conjecture.
\begin{conjecture}
\label{conj:2}
Let $\cF$ be a unitary fusion category (UFC), then there is a completely rational conformal net $\A$ with $\Rep(\A)\cong Z(\cF)$, 
where $Z(\cF)$ is the Drinfeld center.

Equivalently, let $N\subset M$ be a finite depth subfactor, then there is
 a completely rational conformal net $\A$ with $\Rep(\A)\cong D(N\subset M)$,
where $D(N\subset M)$ denotes the quantum double of $N\subset M$.
\end{conjecture}
\begin{rmk}
The net $\A$ in Conjecture \ref{conj:1} or \ref{conj:2}  would be far from unique. Namely, let $\cB$ be a holomorphic net, \ie 
the representation category is trivial $\Rep(\cB)\cong \Hilb$, 
then $\Rep(\A\otimes \cB)\cong\Rep(\A)$. 
\end{rmk} 
So far a technique which produces from a subfactor or a fusion category a conformal field theory
is not established, though see \cite{Jo2014} for some recent approach. 
But subfactors up to index 5 are classified and we can try to exhaust (part of) the classification list, by constructing a CFT
model for every subfactor in the list.  

If we have a UMTC $\cC$ we can replace the braiding by its opposite braiding $\varepsilon(\rho,\sigma)=\varepsilon(\sigma,\rho)^\ast$ which gives (in general) a new UMTC denoted $\rev{\cC}$. 
\begin{conjecture}
\label{conj:3}
Let $\A$ be a completely rational conformal net. Then there exist a completely rational conformal net 
$\tilde \A$, such that $\Rep(\tilde \A)\cong \rev{\Rep(\A)}$.
\end{conjecture}
Here the positivity of energy is crucial. One can easily construct $\tilde \A$ with ``negative energy'' having this property. Note 
that Conjecture \ref{conj:3} would imply that Conjecture \ref{conj:2} would hold for all $\cF=\Rep(\A)$ which are representation 
category of a conformal net $\A$. Indeed, $\cC=\Rep(\A)$ is a UMTC and thus $Z(\cC)\cong \cC \boxtimes \rev{\cC} \cong\Rep(\A\otimes\tilde \A)$.

There are more exotic subfactors for which the realization by conformal field theory in any sense is not know. The first is the Haagerup subfactor
\cite{Ha1994}.
Its quantum double is considered to be exotic in
\cite{HoRoWa2008}.
In the same article also the quantum double of the $E_6$ subfactor is considered exotic. The authors admit that they did not consider simple current extensions.
We show that the double of $E_6$ indeed just arises
as $\ZZ_2$--simple current construction of $\SU(2)_{10}\times \SO(11)_1$ and thus is far from exotic.
We also note that the even part of the $E_6$ subfactor is a pivotal fusion category of rank 
3 and the lowest rank example of a pivotal fusion category which is not braided
by the classification of rank 3 pivotal fusion categories \cite{Os2013}.

Conjecture \ref{conj:2} would give a positive answer to the question: 
\begin{question} 
  \label{quest:AllSubfactors}
  Does every finite depth subfactor come from conformal field theory (\cf 
\cite{Jo2014}).
\end{question}
 Namely, for every completely rational conformal net $\A$, Kawahigashi, Longo, Rehren and the author
 have recently shown
that certain subfactors related to $\Rep(\A)$ classify the phase boundaries of a full conformal field theory
on Minkowski space based on the chiral theory $\A$. 
\begin{prop}[see Proposition \ref{prop:PhaseBoundaries}]
  Let $N\subset M$ be a subfactor and a completely rational conformal net 
  $\A$ with $\Rep(\A)\cong D(N\subset M)$. Then there is a phase boundary related to the subfactor $N\subset M$.
\end{prop}
So, in this sense Conjecture \ref{conj:2} would really give a positive answer to Question \ref{quest:AllSubfactors}.

The main goal of this article is to confirm  Conjecture \ref{conj:2} for the simple case $[M:N]<4$.
\begin{prop}[see Corollary \ref{cor:Doubles}]
  Every quantum double $D(N\subset M)$ of a subfactor $N\subset M$ with  $[M:N] < 4$ is realized by a completely rational conformal net $\A_{N\subset M}$, \ie
  $\Rep(\A_{N\subset M})\cong D(N\subset M)$.
\end{prop}

We note that the next 
possible index is realized by the Haagerup subfactor mentioned above with index 
$$[M:N]=\frac{5+\sqrt{13}}{2} \approx  4.303$$
and there is strong indication in \cite{EvGa2011}, that there is a conformal net realizing its double.
We hope that our techniques here give new ideas to construct this examples.

This article is organized as follows.
In Section \ref{sec:QD} we give some preliminaries about braided subfactors and quantum doubles and in 
 Section \ref{sec:CN} we give some preliminaries about conformal nets on the circle and introduce some examples which we later need. We give some characterization and structural results 
of conformal nets whose representation category is a quantum double.
In Section \ref{sec:Realization} we give some results about conformal nets having  the opposite braiding of a given net. We give 
examples of nets having opposite braiding of $\SU(2)_k$. We give a general criterion how 
the quantum double of a subfactor arising by $\alpha$-induction of an inclusion of conformal nets yields a conformal net 
realizing the quantum double of it. We use these techniques for the realization of quantum doubles for index less than 4
and some sporadic examples between $4$ and $5$.
In Section \ref{sec:VOA}, by using the  categorical nature of our result, whe show how to relate it to vertex operator algebras. 
In particular, there is also are realization 
of quantum doubles of subfactors with index less than 4 by vertex operator algebras.

\subsection*{Acknowledgements}
I would like to thank Zhengwei Liu for some useful comment 
and Luca Giorgetti and Yasuyuki Kawahigashi for remarks on early version of this manuscript.

Ideas for this work was were obtained at the Workshop ID: 1513 ``Subfactors and Conformal Field Theory'' at the 
Mathematische Forschungsinstitut Oberwolfach and the author would like thanks
the organizers and the MFO.

\section{Quantum Doubles}
\label{sec:QD}
We are using here the language of endomorphisms of type III factors (see \cite{BiKaLoRe2014-2}), but the same can be understand in terms of 
bimodules of type II or type III factors or in terms of unitary fusion categories.

We note the it follows from 
\cite{HaYa2000}, (more indirect also from \cite{Po1993,Po1994-2} and in certain cases \cite{Oc1988})
that any abstract unitary fusion category $\cF$ can be realized as $\cF\subset \End(M)$ with $M$ the hyperfinite type III${}_1$ factor. 
By Popa's theorem 
\cite{Po1993}
such a realization is unique, namely if $\tilde\cF\subset \End(N)$ another realization then there is
an isomorphism $N\to M$ implementing the equivalence between the two fusion categories (\cf 
\cite[Proof of Corollary 35]{KaLoMg2001}).

Given an inclusion $N\subset M$ of hyperfinite type III${}_1$ factors $M,N$ with finite minimal index $[M:N] < \infty$
\cite{Jo1983,Ko1986} 
we denote by $\iota\colon N\to M$ the inclusion map. We often write $\iota(N)\subset M$ to have 
a uniform notation if we consider endomorphisms $\rho$ of $M$ and inclusions $\rho(M)\subset M$.
By the finite index assumption, there is a conjugate morphism $\bar \iota\colon M \to N$,
such that $\id_M\prec \iota\circ\bar \iota$ and $\id_N\prec \iota\circ\bar \iota$.
Then $(\bar\iota\circ\iota)^{\circ n}$ and $(\iota\circ\bar\iota)^n$, $n\in \NN$  generate full \Cstar-tensor categories (\cite{LoRo1997})
$\bim[\iota]N\cF N \subset \End_0(N)$  and $\bim[N\subset M]M\cF M \subset \End_0(N)$, respectively,  
and we say that $\iota(N)\subset M$ has finite depth if and only if 
$|\Irr(\bim[N\subset M]N\cF N )|< \infty$, or equivalently
$|\Irr(\bim[N\subset M]M\cF M )|< \infty$.
Similarly, one defines full replete subcategories $\bim[N\subset M]N\cF M=\langle(\bar\iota\circ\iota)^n\circ\bar\iota\rangle\subset \Mor_0(M,N)$
and $\bim[N\subset M]M\cF N=\langle \iota\circ(\bar\iota\circ\iota)^n\rangle\subset \Mor_0(N,M)$.

The (strict) 2-category $\cF^{N\subset M}$ with two 0-objects $\{N,M\}$
and the hom-categories given by $\bim[N\subset M]N\cF N$, $\bim[N\subset M]N\cF M$, $\bim[N\subset M]M\cF N$ and $\bim[N\subset M]M\cF M$ is called the \textbf{standard invariant} of $N\subset M$.
The finite depth condition corresponds to rationality in conformal field theory.

Given a fusion category $\bim N\cF N\subset \End(N)$ and a subfactor $N\subset M$ \textbf{related} to $\bim N\cF N$, \ie $\bar \iota\circ\iota\in \bim N\cF N$
(then $\bim[N\subset M] N{\cF}{N}\subset \bim N\cF N$) the \textbf{dual category}
$\bim M \cF M\subset \End_0(M)$ is the fusion category generated by $\beta\prec \iota \circ\rho \circ\bar\iota$ with 
$\rho \in \bim N{\cF} N$. The categories $\bim N\cF N$ and $\bim M\cF M$ are Morita equivalent
in the sense of  \cite{Mg2003}
the Morita equivalence is given by tensoring with $\iota$ and $\bar \iota$.

We start with a unitary \textbf{modular tensor category (UMTC)} $\bim N\cC N\subset \End_0(N)$,
where the unitary braiding in $\Hom(\rho\circ\sigma,\sigma\circ\rho)$ is denoted by 
$\varepsilon^+(\rho,\sigma)$ or simply $\varepsilon(\rho,\sigma)$ 
and the reversed braiding by $\varepsilon^-(\rho,\sigma)=\varepsilon(\sigma,\rho)^\ast$.

Let us fix $\iota(N)\subset M$ related to $\bim N\cC N$. This gives $\theta=\bar\iota \circ \iota$ the structure 
of an algebra object in $\bim N\cC N$, more precisely a Q-system $\Theta=(\theta,x,w)$. There is a notion of commutativity, namely
let $x\in\Hom(\theta,\theta\circ\theta)$ be the co-multiplication, then
the Q-system is called \textbf{commutative} if and only if $\varepsilon(\theta,\theta)x=x$. 

Let us fix a subfactor $\iota(N)\subset M$ related to $\bim N\cC N$. Then $\alpha$-induction maps from $\cC=\bim N\cC N$ to the dual category $\cD =\bim M\cC M$ and is given by:
\begin{align*}
  \bim N \cC N &\longrightarrow \bim M \cC M\subset \bim M \cC M\\
  \lambda & \longmapsto \alpha ^\pm_\lambda :=\bar\iota^{-1}\circ \Ad (\varepsilon^\pm(\lambda,\theta))  
  \circ \lambda\circ \bar \iota \in \End(\M)
  \,.
\end{align*}
We denote by $\cD_\pm \equiv\bim[\pm]M{\cC}M= \langle \alpha_\rho^\pm:\rho\in\bim N\cC N\rangle$ the UFC
generated by $\alpha^\pm$-induction, respectively, 
and  by $\cD_0\equiv\bim[0]M{\cC}M=\cD_+\cap\cD_-$ the \textbf{ambichiral} category.

Let $\cF$ be a unitary fusion category. We can assume that it is (essentially uniquely) 
realized as $\cF\cong\bim N \cF N\subset \End_0(N)$
with $N$ a hyperfinite type III${}_1$ factor.
Let 
$A:=N\otimes N^\op \subset \iota_\LR(B)$ be the Longo--Rehren inclusion
with $\bim A\cC A \cong \bim N \cF N \boxtimes \bim[\op] N \cF N$
and $\bim B \cC B$ 
the category generated 
by $\langle (\bar\iota_\LR\circ \beta\circ\iota_\LR)^n : \beta \in \bim A\cC A \rangle \subset \End(B)$.
Then Izumi showed that $\bim B \cC B\cong Z(\cF)$, where $Z(\cF)$ denotes the unitary Drinfeld center
\cite[Section 6]{Mg2003II} of
$\cF$, which is a UMTC by \cite{Mg2003II}.
The Q-system $\Theta_\LR=(\theta_\LR,w_\LR,x_\LR)$ with $\theta=\bar \iota_\LR\circ\iota_\LR$ is  commutative and $d \theta_\LR = \Dim(\cF)$,
where $\Dim(\cF) = \sum_{\rho\in \Irr(\cF)} d\rho^2$ is the \textbf{global dimension}.

If we start with a finite depth subfactor $\iota(N)\subset M$, then 
$Z(\bim[N\subset M]N\cF N)\cong  Z(\bim[N\subset M]M\cF M)$ (see proof of Proposition \ref{prop:sandwich} below) and we can talk about the 
\textbf{quantum double of} $\iota(N)\subset M$, denoted by $D(N\subset M)$.

\begin{example}
The quantum double $D(N\subset M)$
has been calculated in
\cite[Section 4]{EvKa1998} for $A_n$ subfactors and \cite[Examples 5.1,5.2]{BcEvKa2001} for $E_6$ and $E_8$ subfactors. 
The quantum double of $E_6$ has also been computed using the tube algebra 
and half-braidings in \cite{Iz2001II}.
\end{example}

The quantum double is related to the Ocneanu's asymptotic inclusion
\cite{Oc1988}, Popa's symmetric enveloping algebra
\cite{Po1994}
and the Longo--Rehren subfactor
\cite{LoRe1995}, see also 
\cite{Ma2000,Iz2000}.

Izumi showed \cite{Iz2000}  that there is a Galois correspondence, namely there is a one-to-one correspondence between intermediate subfactors
$B\subset Q \subset A$ and subcategories $\cG \subset \cF$.

The following (3) was observed \cite[Theorem 12]{Oc2001} for $\cC$ being a $\SU(2)_k$ category and
is partially contained in \cite[Corollary 3.10, 4.8]{BcEvKa2001}.
\begin{prop}
  \label{prop:sandwich}
  Let $\cC \subset \End(N)$ be a UMTC, and $\iota(N)\subset M$ subfactor with commutative
  Q-system $\Theta\in \cC$. 
Denote by $\cD=\langle \beta \prec \iota \rho\bar \iota : \rho\in\cC\rangle\subset \End(M)$ the dual category. Then 
\begin{enumerate}
  \item $Z(\cC)\cong Z(\cD ) \cong \cC \boxtimes \rev{\cC}$,
  \item $Z(\cD_0) \cong \cD_0\boxtimes \rev{\cD_0}$,
  \item $Z(\cD_+) \cong \cC\boxtimes \rev{\cD_0}$,
  \item $Z(\cD_-) \cong \rev{\cC}\boxtimes \cD_0$.
\end{enumerate}
\end{prop}
\begin{proof}
For (1) it follows by \cite{Sc2001} together with \cite{Mg2003II} that $Z(\cC)\cong Z(\cD)$, because $\cC$ and $\cD$ are Morita equivalent, and again by  \cite{Mg2003II}  $Z(\cC)\cong \cC\boxtimes \rev{\cC}$.
It was shown \eg in \cite[Theorem 4.2]{BcEvKa2000}, that $\cD_0$ is modular, thus the statement (2) follows from (1).
$\cD_+$ is equivalent with $\cC_\Theta$ (\cf \cite[Remark 5.6]{BiKaLo2014}) and 
by \cite[Corollary 3.30]{DaMgNiOs2013}, see also \cite[Remark 4.3]{DrGeNiOs2010} we have 
$Z(\cC_\Theta)\cong \cC\boxtimes \rev{\cC^0_\Theta}$, which is braided equivalent with $\cC\boxtimes\cD_0$, thus (3).
Finally, (4) follows by applying (3) to $\rev{\cC}$.
\end{proof}

\section{Conformal Nets}
\label{sec:CN}
By a conformal net $\A$, we mean a local M\"obius covariant net on the circle. 
It associates with every 
proper interval $I\subset S^1\subset \CC$ on the circle a von Neumann algebra $\A(I)\subset \B(\Hil_\A)$
on a fixed Hilbert space $\Hil$, such that the following properties hold:
    \begin{enumerate}[{\bf A.}]
        \item \textbf{Isotony.} $I_1\subset I_2$ implies 
            $\A(I_1)\subset \A(I_2)$.
        \item \textbf{Locality.} $I_1  \cap I_2 = \emptyset$ implies 
            $[\A(I_1),\A(I_2)]=\{0\}$.
        \item \textbf{Möbius covariance.} There is a unitary representation
            $U$ of $\Mob$ on $\Hil$ such that 
            $  U(g)\A(I)U(g)^\ast = \A(gI)$.
        \item \textbf{Positivity of energy.} $U$ is a positive energy 
            representation, i.e. the generator $L_0$ (conformal Hamiltonian) 
            of the rotation subgroup $U(z\mapsto \e^{\ima \theta}z)=\e^{\ima \theta L_0}$
            has positive spectrum.
        \item \textbf{Vacuum.} There is a (up to phase) unique rotation 
            invariant unit vector $\Omega \in \Hil$ which is 
            cyclic for the von Neumann algebra $\A:=\bigvee_{I\in\cI} \A(I)$.
    \end{enumerate}
A local Möbius covariant net on $\A$ on $\Sc$ is called \textbf{completely 
rational} if it 
\begin{enumerate}[{\bf A.}]
  \setcounter{enumi}{5}
  \item fulfills the 
    \textbf{split property}, \ie 
    for $I_0,I\in \cI$ with $\overline{I_0}\subset I$ the inclusion 
    $\A(I_0) \subset \A(I)$ is a split inclusion, namely there exists an 
    intermediate type I factor $M$, such that $\A(I_0) \subset M \subset \A(I)$.
  \item is 
 \textbf{strongly additive}, \ie for $I_1,I_2 \in \cI$ two adjacent intervals
obtained by removing a single point from an interval $I\in\cI$ 
the equality $\A(I_1) \vee \A(I_2) =\A(I)$ holds.
  \item for $I_1,I_3 \in \cI$ two intervals with disjoint closure and 
    $I_2,I_4\in\cI$  the two components of $(I_1\cup I_3)'$, the 
    \textbf{$\mu$-index} of $\A$
    \begin{equation*}
      \mu(\A):= [(\A(I_2) \vee \A(I_4))': \A(I_1)\vee \A(I_3) ]
    \end{equation*}
    (which does not depend on the intervals $I_i$) is finite.
\end{enumerate}

A \textbf{representation} $\pi$ of $\A$ is a family of representations 
$\pi=\{\pi_I\colon\A(I)\to \B(\Hil_\pi)\}_{I\in\cI}$ on a common Hilbert space $\Hil_\pi$ which are compatible, i.e.\  $\pi_J\restriction \A(I) =\pi_I$ for $I\subset J$.
Every non-degenerate representation $\pi$ with $\Hil_\pi$ separable 
turns---for every choice of an interval $I_0\in\cI$---out to be equivalent to a representation 
$\rho$ on $\Hil$, such that $\rho_J=\id_{\A(J)}$ for $J\cap I_0=\emptyset$. Then Haag duality implies
that $\rho_{I}$ is an endomorphism of $\A(I)$ for every $I \in \cI$ with $I\supset I_0$. 
Thus we can realize the representation category of $\A$ inside the C$^\ast$ tensor category of endomorphisms $\End_0(N)$
of a type III factor $N=\A(I)$ and the embedding turns out to be full and replete. We denote this category by $\Rep^I(\A)$.
In particular, this gives the representations of $\A$ the structure of a tensor category \cite{DoHaRo1971}. It has a natural \textbf{braiding}, which is completely fixed
by asking that if $\rho$ is localized in $I_1$ and $\sigma$ in $I_2$ where $I_1$ is left of $I_2$ inside $I$ then $\varepsilon(\rho,\sigma)=1$
 \cite{FrReSc1989}. The \textbf{statistical dimension} of $\rho\in\Rep^I(\A)$ is given by
$d\rho=[N:\rho(N)]^{\frac12}$.
Let $\A$ be completely rational conformal net, then by \cite{KaLoMg2001} $\Rep^I(\A)$ is a UMTC
and $\mu_\A=\dim(\Rep^I(\A))$.

We write $\A\subset \cB$ or $\cB\supset \A$
if there is a representation $\pi=\{\pi_I\colon\A(I)\to\cB(I)\subset\B(\Hil_\cB)\}$ of $\A$ on 
$\Hil_\cB$  and an isometry $V\colon \Hil_\A\to \Hil_\cB$
with $V\Omega_\A=\Omega_\cB$ and $VU_\A(g)=U_\cB(g)V$.
We ask that further that
$Va=\pi_I(a)V$ for $I\in\cI$, $a\in\A(I)$. Define $p$ the projection 
on $\Hil_{\A_0}=\overline{\pi_I(\A(I))\Omega}$. Then $pV$ is a unitary equivalence
of the nets $\A$ on $\Hil_\A$ and  $\A_0$ defined by $\A_0(I)=\pi_I(\A(I))p$ on $\Hil_{\A_0}$.

\begin{defi}
\label{def:coset}
Let 
$\A\subset \cB$ an inclusion of conformal nets. Then we define the 
\textbf{coset net} 
$\A^\mathrm{c}(I)=\cB(I)\cap \A'$.
Note that $\A^\mathrm{c}\subset \cB$. We call $\A\subset \cB$ \textbf{normal} if 
$\A^\mathrm{cc}=\A$.
We call $\A\subset \cB$ \textbf{co-finite} if $[\cB(I):\A(I)\otimes\A^\mathrm{c}(I)]<\infty$.
\end{defi}
For every co-finite extension $\A\subset \cB$ holds:
$\cB$ is completely rational iff $\A$ and $\A^\mathrm{c}$ are completely
rational \cite{Lo2003}.

\subsection{On conformal nets realizing quantum doubles/Drinfeld centers}
In this section we give some structural results about conformal nets whose representation
category is a quantum double.
If we talk about a subfactor $N\subset M$, we are just interested in finite depth subfactors 
which are hyperfinite of type II$_{1}$ or III$_{1}$. In this case standard invariant is a complete invariant \cite{Po1993}.
We might also replace subfactor by subfactor standard invariant.
We write $N\subset M \approx N_1\subset M_1$ if both have equivalent standard invariant.

\begin{defi}
A \textbf{holomorphic net} $\A$ is a completely rational conformal net with trivial representation category $\Rep(\A)\cong\Hilb$,
or equivalently \cite{KaLoMg2001} with $\mu_\A=1$.
\end{defi}
\begin{prop}[{\cf \cite[Corollary 3.5]{Mg2010}, \cite[Theorem 2.4]{Ka2015}}] 
  \label{prop:Holomorphic}
  Let $\A$ be a completely rational conformal net. The following are equivalent:
  \begin{enumerate}
    \item There is a holomorphic local irreducible extension $\cB\supset \A$. 
    \item $\Rep(\A) \cong Z(\cF)$ for some unitary fusion category $\cF$.
    \item $\Rep(\A) \cong D(N\subset M)$ for some finite depth subfactor $N\subset M$.
  \end{enumerate}
\end{prop}
\begin{proof} Given $N\subset M$ take $\cF:=\bim N\cF N$. Conversely, we may assume 
that $\cF$ is a full subcategory of $\End(M)$ and we can take $N=\rho(M)\subset M$,
where $\rho=\bigoplus_{\rho_i\in\Irr(\cF)} \rho_i$. Thus (2) and (3) are equivalent.

If (2) is true the dual Q-system of the Longo--Rehren inclusion associated with $\cF$ gives 
a commutative Q-system $\Theta=(\theta,w,x)$ in $\Rep^I(\A)$ with $d\theta=\sqrt{\mu_\A}$ 
the corresponding extension $\cB\supset \A$ has $\mu_\cB=1$. 

Conversely, if (1) holds, let $\Theta=(\theta,w,x)$ in $\Rep^I(\A)$ be the Q-system characterizing 
$\cB\supset \A$. The Q-system $\Theta$ is commutative with $d\theta=\sqrt{\Dim \Rep \A}$,
thus a Lagrangian Q-system
which forces $\Rep(\A)\cong Z(\cF)$ for some fusion category $\cF$.

Indeed, for $N:=\A(I)\subset \cB(I):=M$ and $\bim N\cC N=\Rep^I(\A)$ 
 using Proposition \ref{prop:sandwich} (3) 
we get 
$$Z(\cF)=Z(\bim[+]M\cC M)\cong \bim N\cC N\boxtimes \bim[0] M\cC M\cong
  \Rep^I(\A) \boxtimes \Rep^I(\cB) \cong \Rep^I(\A)\,$$
using \cite[Proposition 6.4]{BiKaLo2014} in the second last step.
\end{proof}
\begin{rmk} One might see $\A\subset \cB$ as a generalization of an orbifold by a finite group. 
Namely, if 
$\cF$ is pointed and the fusion rules are given by the finite group $G$, 
then for the associated with $\cF$ associated extension $\cB\supset \A$ from Proposition \ref{prop:Holomorphic} the net $\A =\cB^G$  
is indeed the $G$-orbifold of $\cB$, \ie $\A=\cB^G$, \cf \cite{Mg2005}.
\end{rmk}

Let $N\subset M$ be a finite index and finite depth subfactor. 
 Conjecture \ref{conj:2} is equivalent with the existence of a conformal net $\A$ with 
$\Rep(\A) \cong D(N\subset M)$ for every such $N\subset M$.
Conversely, in the following Proposition we show that if such a net $\A$ exists, there are two extensions 
$\cB_{N}$ and $\cB_{M}$, such that $\cB_{N}(I)\subset\cB_{M}(I)\approx
N\subset M$. But any morphism 
$\beta \colon \cB_{N}(I) \to \cB_{M}(I)$ related to 
$\Rep^I(\A)$, \ie $\bar\iota_{\cB_{M}(I)}
 \circ \beta\circ \iota_{\cB_{N}(I)}\in \Rep^I(\A)$, prescribes a defect line or phase boundary \cite{BiKaLoRe2014} between 
the full conformal field theories $\cB_\mathrm{L}=\cB_{N}\otimes \cB_{N} \supset \A\otimes \A$ and
 $\cB_\mathrm{R}\supset \A\otimes \A$ on 2D Minkowski space, which is invisible if restricted to $\A\otimes \A$,
also called $\A$--topological. Here the net $\cB_R$ comes  from the $\alpha$-induction 
construction \cite{Re2000} of $\A\subset \cB_M$, which coincides with the full center construction \cite{BiKaLo2014}.
Thus  the subfactor $\cB_{N}(I)\subset \cB_{M}(I)\approx N\subset M$ is related to a phase boundary in conformal field theory.
\begin{prop}
  \label{prop:PhaseBoundaries}
  Let $\A$ be a completely rational net with $\Rep(\A)\cong D(N\subset M)$.
  Then there exist $\cB_{N} \supset \A$ local extension with
   $\Rep(\cB_\bullet)\cong \Hilb$ and a (non-local) extension $\cB_{M}\supset \cB_{N}\supset \A$  
  with 
  $\cB_N(I)\subset \cB_M(I) \approx N\subset M$.
  Thus the inclusion $\iota \colon \cB_{N}(I) \to \cB_{M}(I)$ is related to 
  $\Rep^I(\A)$ and prescribes a phase boundary in the sense of \cite{BiKaLoRe2014}.
\end{prop}
\begin{proof} The dual Q-systems of the Longo-Rehren inclusion associated with 
  $\bim[N\subset M]N\cF N$ gives a commutative Q-systems
  in $D(N\subset M)\cong Z(\bim[N\subset M]N\cF N)\cong Z(\bim[N\subset M]M\cF M)$, which
  we use to define the local extensions $\cB_N\supset\A$.
  Let $A=\A(I)$, $B_N= \cB_N (I)$, then with $\bim A\cC A\cong D(N\subset M)$ 
  we have $\bim {B_N}\cC{B_N} \cong\bim[N\subset M] N\cF N\boxtimes(\bim[N\subset M] N\cF N)^\op$.
  Finally, the Q-system $\Theta_{N\subset M}\boxtimes \id =(\theta_{N \subset M}\boxtimes \id,w_{N\subset M}\boxtimes 1_{\id}, x_{N\subset M} \boxtimes 1_{\id})$ gives an extension $B_M\supset B_N$ which gives a non-local extension $\cB_M\supset \A$,
  where $\Theta_{N\subset M}$ is the Q-system in $\bim[N\subset M]N\cF N$ of the subfactor $N\subset M$.
  Because  $B_N\subset B_N$ and $N\subset M$ have by construction equivalent Q-systems, they have the same standard invariant.
\end{proof}

\subsection{Some conformal nets}
\begin{example}
We denote by $\A_{\SU(2),k}$ or simply by $\A_k$ the $\SU(2)$ loop group net at level $k$ \cite{Wa}, which is completely rational \cite{Xu2000}
and thus gives a UTMC $\Rep(\A_k)$. 
The simple objects are $\{\rho_0,\ldots,\rho_k\}$ with fusion rules
$$
  [\rho_i]\times [\rho_j]=\bigoplus_{\substack{\ell=|i-j|\\i+\ell \text{ even}\\i+j+\ell\leq 2k}}^{i+j}[\rho_\ell].
$$
The dimensions $d\rho_i$ and twists $\omega_{\rho_i}$ are given by
\begin{align*}
  d_i&=d\rho_i =[i+1]_q:=\frac{\sin
    \frac{(i+1)\pi}{k+2}
    }{\sin 
    \frac{\pi}{k+2}
    },
  &
  \omega_i&=\omega_{\rho_i}=\exp\left({2\pi\ima\frac{i(i+2)}{4(k+2)}}\right),&
  q&=\exp\left({\frac{\ima\pi}{k+2}}\right)
\end{align*}
and the central charge $c_k$ and global dimension $D_k$ by
$$
  c_k=\frac{3k}{k+2}, \quad
  D_k=\sum_{i=0}^{k} d_i^2 =\frac{k+2}{2\sin^2 \left(\frac{\pi}{k+2}\right)}\,.
$$
\end{example}

We remember the classification of $\SU(2)_k$ conformal nets \cite{KaLo2004}, \cite{BcEv1998}.
\begin{prop}
\label{prop:ADE}
Local irreducible extensions $\cB\supset \A_{k}$,
\ie a local net $\cB$ containing $\A_k$ as a subnet,
such that $\A_k(I)'\cap \cB(I)=\CC$
are in one-to-one correspondence with 
\ADEe Dynkin diagrams of Coxeter number $k+2$.
The $E_{6,8}$ Dynkin diagram correspond to the conformal inclusions
 $\A_{10}\subset \A_{\Spin(5),1}$ and $\A_{28}\subset \A_{\Gtwo,1}$, respectively.

The subfactor $\alpha^\pm_{\rho_1}(\cB(I))\supset \cB(I)$ has a principal graph the corresponding 
Dynkin diagram. 
\end{prop}

\begin{example}
  The loop group net of $\Spin(2n+1)$ at level 1 $\A_{\Spin(2n+1),1}$ \cite[Theorem 3.1]{Bc1996} and \cite[Lemma 3.1]{Xu2009} has global dimension $D=4$ 
  and  
  has the Ising fusion rules, \ie the same fusion rules as the net $\A_{\SU(2),2}=\A_{\Spin(3),1}$. 
  We denote the (choice of simple objects) by $\{\rho_0,\rho_1,\rho_2\}$. The category 
  is determined by the fusion rules and twists \cite[Proposition 8.2.6]{FrKe1993}, which are:
$$
  \omega_{\rho_1}=\exp\left({\frac{2\pi \ima(2n+1)}{16}}\right),\quad
  \omega_{\rho_2}=-1\,.
$$
\end{example}

\begin{example}
We get a net $\A_{\Gtwo,1}$ associated with $(\Gtwo)_1$ as an extension of $\A_{28}$.
The category of representations is the Fibonacci or golden category with fusion rules $[\tau]\times[\tau]=[\id]+[\tau]$.

There is a conformal inclusion of $\A_{\SU(3),2}\otimes \A_{\SU(3),1}\subset \A_{\Ffour,1}$,
thus $\A_{\Ffour,1}$ is completely rational. There is also $\A_{\Ffour,1}\otimes \A_{\Gtwo,1}\subset \A_{\grE_8,1}$,
in particular $\Rep(\A_{\Gtwo,1})\cong \rev{\Rep(\A_{\Ffour,1})}$, which is an application of Proposition \ref{prop:Mirror}.
\end{example}

\begin{example}
For central charge $c$ with values 
$$  
  c_m=1-\frac6{(m+1)(m+2)}, \quad (m=2,3,\ldots)\,,
$$ 
the \textbf{Virasoro net} $\Vir_{c_m}$ is given by  the coset net of the inclusion $\A_{m}\subset \A_{m-1}\otimes\A_1$
\cite{KaLo2004}, in other words we have the conformal inclusion
$$
  \Vir_{c_m}\otimes \A_{m}\subset \A_{m-1}\otimes\A_1\,.
$$
and $\Vir_{c_m}$ is completely rational, see \cite{KaLo2004}.
\end{example}

\section{Realization of some Quantum Doubles by Conformal Nets}
\label{sec:Realization}
\subsection{Realization of the opposite braiding}
\begin{prop}
  \label{prop:MirrorExtension}
  Let $\A,\tilde \A$ be completely rational conformal nets with $\Rep(\tilde \A)\cong \rev{\Rep(\A)}$
  and $\cB\supset\A$ be an irreducible local extension (which is automatically completely rational).
  Then there is an irreducible local extension
  $\tilde \cB\supset \tilde \A$ with $\Rep(\tilde\cB)\cong \rev{\Rep(\cB)}$. 
\end{prop}
\begin{proof} Using the equivalence $\Rep(\tilde \A)\cong \rev{\Rep(\A)}$, the commutative Q-system $\Theta \in\Rep(\A)$ gives a 
  commutative Q-system $\tilde \Theta\in \Rep(\tilde \A)$, which defines an extension
  $\tilde \cB \supset \tilde \A$ with the asked properties.
\end{proof}
\begin{rmk} This is a trivial instance of mirror extensions \cite{Xu2007}, namely
take $\cB_{\LR}\supset \A\otimes \tilde \A$ the Longo--Rehren extension \cite{LoRe1995}, which gives
$\Rep(\cB_{\LR})\cong \Hilb$.  Then $\A \subset \cB_{\LR}$ is
normal and co-finite and $\tilde \A$ is its coset and $\tilde \cB\supset \tilde \A$
the mirror extension of $\A\subset \cB$. Using \cite[Proposition 6.4]{BiKaLo2014} 
$\Rep(\tilde\cB)$ is equivalent as UFC with $\Rep(\cB)$ and has the opposite braiding.
\end{rmk}

\begin{prop}
\label{prop:Mirror}
 Let $\cB$ be a holomorphic net and $\A\subset \cB$ be co-finite and normal.
Let $\tilde \A$ be the coset net of the inclusion $\A\subset \cB$.
Then the nets $\A$ and $\A^\mathrm{c}$ are completely rational with 
$\Rep(\A^\mathrm{c})\cong \rev{\Rep(\A)}$.
\end{prop}
\begin{proof}
$\A$ and $\A^\mathrm{c}$ are completely rational by assumption (using \cite{Lo2003} see above). 

The Q-system $\Theta=(\theta,w,x)$ giving the extension $\A\otimes \A^\mathrm{c} \subset \cB$
is of the form 
$$
  [\theta]=\bigoplus_{\substack{\mu\in\Irr(\Rep(\A))\\\nu\in\Irr(\Rep(\A^\mathrm{c}))}} 
  Z_{\mu,\nu} [\mu\otimes \nu]\,.
$$
By normality of $\A,\A^{c}\subset \cB$ we have $Z_{\mu,\id}=\delta_{\id,\mu}$ 
  and $Z_{\id,\nu}=\delta_{\id,\nu}$.
Then it  
follows that there is a braided equivalence $\phi\colon \cC\to \rev{\cD}$, 
for some full and replete subcategories $\cC\subset \Rep(\A)$ and $\cD\subset \Rep(\cB)$, such that $\Theta$
is the by $\phi$ twisted Longo--Rehren extension, see \cite[Definition 4.1]{BiKaLo2014} for the definition.
On the one hand  $(d\theta)^2=\Dim\Rep(\A)\cdot\Dim\Rep(\A^\mathrm{c})$, because   $\cB$ is holomorphic.
On the other hand 
$d\theta=\Dim\C=\Dim \cD$. Together, because all dimensions are positive, this implies $\cC=\Rep(\A)$ and $\cD= \Rep(\A^\mathrm{c})$.
\end{proof} 
Let $\A_k=\A_{\SU(2),k}$ and 
let $\cB_k$ be the coset net of 
$$
  \A_{k} \subset \A_{1}^{\otimes k}=\A_{1}^{\otimes k}\,
$$
which is normal by \cite[Lemma 4.2 (1)]{Xu2007}.
By induction, it follows that we have conformal inclusions:
$$
\A_k \otimes \Vir_{c_2}\otimes \cdots \otimes \Vir_{c_k}\subset
  \A_k \otimes \cB_k  \subset
 \A_{\SU(2),1}^{\otimes k}
$$
thus $\cB_k$ it is completely rational by \cite{Lo2003}.  
Using the conformal inclusion $\A_{\grE_7,1}\otimes \A_1\subset \A_{\grE_8,1}$,
 which are all conformal nets associated with 
even lattices (\cf \cite{Bi2012}) and 
which is a Longo--Rehren extension and thus normal
we get the conformal inclusions:
$$
  \A_{k}\otimes \cB_k \otimes \A_{\grE_7,1}^{\otimes k} \subset \A_{A_1}^{\otimes k}\otimes \A_{\grE_7,1}^{\otimes k} \subset \A_{\grE_8,1}^{\otimes k}\,.
$$
Now we take $\tilde \A_{k}$ to be the coset of the normal inclusion \cite[Lemma 4.2 (1)]{Xu2007} $\A_{k}\subset \A_{\grE_8,1}^k$. This is completely rational, because it is an intermediate net of completely rational nets:
$$\A_{k}\otimes \cB_k \otimes \A_{\grE_7,1}^{\otimes k} \subset \A_k\otimes \tilde \A_k  \subset \A_{\grE_8,1}^{\otimes k} \,.$$
Thus using Proposition \ref{prop:Mirror} we have proven:
\begin{prop} 
  \label{prop:Akmirror}
  The coset net $\tilde \A_k$ of the inclusion $\A_k\subset \A^{\otimes k}_{\grE_8,1}$ above is completely rational
  with $\Rep(\tilde \A_k)\cong \rev{\Rep(\A_k)}$.
\end{prop}

\begin{example}
We note that $\tilde \A_1=\A_{\grE_7,1}$ and that  
  $\Vir_{c_k}\otimes \tilde \A_1\otimes \tilde \A_{k-1}\otimes \A_k\subset \A_{\grE_8,1}^{\otimes k}$. We get the intermediate inclusion:
$$
  \Vir_{c_k}\otimes \tilde \A_1\otimes \tilde \A_{k-1}\otimes \A_k\subset \tilde\A_k\otimes \A_k \subset \A_{\grE_8,1}^{\otimes k} 
$$
Thus also $\Vir_{c_k}\otimes\tilde \A_{k-1}\otimes\tilde\A_1\subset\tilde \A_k$ and $\Vir_{c_k}$ can be obtained 
back from the coset of $\tilde\A_{k-1}\subset \tilde\A_1\otimes \tilde\A_k$.
We also get 
that $\Vir_{c_m}\subset \A_{\grE_8,1}^{\otimes m}$ is normal and co-finite, and thus its coset
$\tilde \Vir_{c_m}=\Vir_{c_m}^\mathrm{c}$ realizes the opposite braiding of $\Vir_{c_m}$.
Further, $\Vir_{c_m}\otimes\tilde\Vir_{c_m}$ realizes, using Proposition \ref{prop:sandwich}(1), the Drinfeld center $Z(\Rep^I(\Vir_{c_m}))$.
\end{example}

\subsection{Realization of quantum doubles}
The next proposition shows, that if a subfactor $N\subset M$ arises from $\alpha$-induction of 
a local irreducible extension $\A\subset \cB$ and we have a net $\tilde \A$ realizing 
the opposite braiding of $\A$, then the there is a net $\cB_{N\subset M}$ with 
$\Rep(\cB_{N\subset M})=\cD(N\subset M)$. 
\begin{prop} 
  \label{prop:GaloisDouble}
  Let $N\subset M$ be an irreducible subfactor.  
  Assume there exists  a completely rational conformal net $\A$ and an irreducible local extension 
  $\cB\supset \A$,
  such that $N\subset M$ arises by $\alpha^\pm$ induction, \ie
  there is a  $\rho\in \bim{\A(I)}{\cC}{\A(I)}=\Rep^I(\A)$
  and a $[\beta]\prec[\alpha_\rho^{\pm}]$,
  such that $\beta(\cB(I))\subset \cB(I) \sim N\subset M$. 
Further, assume there exists $\tilde \A$, a completely rational conformal net with $\Rep(\tilde \A)\cong \rev{\Rep(\A)}$.
Then 
\begin{enumerate}
  \item There exists a completely rational conformal net $\cB_{N\subset M}$ realizing the 
    quantum double $D(\N\subset M)$, \ie $\Rep(\cB_{N\subset M})\cong D(N\subset M)$.
  \item
    It can be given as a local irreducible extension: 
   \begin{itemize}
      \item $\cB_{N\subset M}\supset \A\otimes \tilde\cB$,
  in the $\alpha^+$ case 
  or
    \item  $\cB_{N\subset M}\supset \tilde \A \otimes \cB$ in the $\alpha^-$ case.
    \end{itemize}
  \item In the case that $[\beta],[\bar\beta]$ (tensor) generate $\bim[\pm]{\cB(I)}\cC{\cB(I)}$, 
 but $[\bar\beta\circ\beta]$ does not,
  the former extension is a $\ZZ_2$--simple current extension.
  \item In the case that $[\bar\beta\circ\beta]$  (tensor) generates $\bim[\pm]{\cB(I)}\cC{\cB(I)}$, then
   $\cB_{N\subset M}$ equals 
   $\A\otimes \tilde \cB$ or $\tilde \A \otimes \cB$, respectively.
\end{enumerate}
\end{prop}
\begin{proof}
  By Proposition \ref{prop:sandwich} we have
  $\Rep(\A\otimes \tilde\cB) \cong Z(\bim[+]{\cB(I)}\cC{\cB(I)})$ 
  and
  $\Rep(\tilde \A\otimes \cB) \cong Z(\bim[-]{\cB(I)}\cC{\cB(I)})$.
  
  Let  $\bim[\beta]{\cB(I)}\cC{\cB(I)}\subset  \bim[\pm]{\cB(I)}\cC{\cB(I)}$ be the subcategory (tensor) generated 
  by $\bar\beta\circ\beta$, then by $D(N\subset M)\cong Z( \bim[\beta]{\cB(I)}\cC{\cB(I)})$ by assumption.

  Further, there is a holomorphic net $\cB_\mathrm{holo} \supset \A\otimes \tilde\cB$
  or
  $\cB_\mathrm{holo} \supset  \tilde \A\otimes \cB$, respectively,  
  which is the Longo--Rehren inclusion and by 
  Galois correspondence there is an intermediate net  $\cB_{N\subset M}$
  with $\Rep(\cB_{N\subset M})\cong Z(\bim[\beta]{\cB(I)}\cC{\cB(I)})
  \cong D(N\subset M)$.

  In the case of (2) we have $2\Dim   \bim[\beta]{\cB(I)}\cC{\cB(I)} = \Dim \bim[\pm]{\cB(I)}\cC{\cB(I)} $
  and 
  $\cB_\mathrm{holo} \supset \A\otimes \tilde\cB$
  or
  $\cB_\mathrm{holo} \supset  \tilde \A\otimes \cB$, respectively,  
  have index two, thus it is a $\ZZ_2$--simple current extensions.

  In the case of (3) we have $\bim[\beta]{\cB(I)}\cC{\cB(I)} =\bim[\pm]{\cB(I)}\cC{\cB(I)}$, respectively
  and the extension is trivial.
\end{proof}

For subfactors with index $<4$ it is well-known that they arise via
$\alpha$-induction from $\SU(2)_k$ loop group models $\A_k$, see Proposition \ref{prop:ADE}.  Together 
with $\tilde \A_k$ from Proposition \ref{prop:Akmirror}
we thus get:
\begin{cor}
  \label{cor:Doubles}
  For every subfactor $N\subset M$ with $[M:N]<4$, \ie 
  for every the standard invariant label by $G\in\{A_n,D_{2n},E_{6,8},\bar E_{6,8}\}$, there is a conformal net 
  $\A_{N\subset M}$
  with $\Rep(\A_{N\subset M})=D(N\subset M)$. The realizations can be given as follows:
  \begin{description}
  \item[$A_{k+1}$] $\A_{k}\otimes \tilde \A_k \myrtimes_{\rho_{k,k}} \ZZ_2$
    the simple current extension with respect to the automorphism $\rho_k\otimes\tilde\rho_k$.
  \item[$D_{2n}$] $\cB_{D_{2n}}\otimes \tilde\cB_{D_{2n}}$, where 
    $\cB_{D_{2n}}$ and $\tilde \cB_{D_{2n}}$ are the $\ZZ_2$--simple current extensions
    of $\A_{4n-4}$ and $\tilde \A_{4n-4}$ by $\rho_{4n-4}$ and $\tilde \rho_{4n-4}$, respectively.
  \item[$E_6$] $\A_{10}\otimes \A_{\Spin(11),1} \myrtimes_{[\rho_{10,2}]} \ZZ_2$, where we can 
    replace $\A_{\Spin(11),1}$ by $\tilde\cB\supset\tilde \A_{10}$, the 
        extension obtained from Proposition \ref{prop:MirrorExtension} applied to $\A_{10}\subset  \A_{\Spin(5),1}$.
  \item[$\bar E_6$] $\tilde \A_{10}\times \A_{\Spin(5),1} \myrtimes_{[\rho_{10,2}]} \ZZ_2$.
  \item[$E_8$] $\A_{28}\otimes \A_{\Ffour,1} \myrtimes_{[\rho_{28,0}]} \ZZ_2$,
    \ie it is given by $\cB_{D_{16}}\otimes \A_{\Ffour,1}$. We 
       can replace $\A_{\Ffour1}$ by the $\tilde\cB\supset\tilde \A_{28}$, the 
        extension obtained from Proposition \ref{prop:MirrorExtension} applied
      to $\A_{28}\subset  \A_{\Gtwo,1}$. 
  \item[$\bar E_8$] 
     $\tilde \A_{28}\otimes \A_{\Gtwo,1} \myrtimes_{[\rho_{28,0}]} \ZZ_2$
    \ie it is given by $\tilde \cB_{D_{16}}\otimes \A_{\Gtwo,1}$.
  \end{description}
\end{cor}
\begin{proof} 
  All subfactors arise as 
  $\alpha^{\pm}_{\rho_1}(\cB_G(I))\subset \cB_G(I)$, where $\cB_G\supset \A_k$ 
  is the extension in
  Proposition \ref{prop:ADE}.
  Further $\bar \alpha_{\rho_1}^\pm$
  generates $\bim[\pm]{\cB(I)}{\cC}{\cB(I)}$, while $\bar\alpha^\pm_{\rho_1}\circ\alpha_{\rho_1}^\pm$
  does not. Thus in  each case we are in the situation of case (2) 
  of Proposition  
  \ref{prop:GaloisDouble}
  and in each case there is just one possible $\ZZ_2$--simple current extension.
\end{proof}

\begin{rmk}
Our method also applies to some subfactors with index between $4$ and $5$:
\begin{itemize}
  \item The  GHJ subfactor \cite{GoHaJo1989}
  with index $3+\sqrt{3}$ arises as the subfactor $\A_{10}(I)\subset \A_{\Spin(5),1}(I)$, 
  see \cite[Section 2.2]{BcEvKa1999}. Thus the even part
  of it coincides with the even part of $\Rep(\A_{10})$, \ie with the even part of 
  the $A_{11}$ subfactor. Thus its quantum double is the same as of the $A_{11}$ subfactor
  and thus also realized by $\A_{10}\otimes\tilde \A_{10}\rtimes_{\rho_{10,10}} \ZZ_2$.
  \item The 2221 subfactor with index $(5+\sqrt{21})/2$ arises from the conformal inclusion $\A_{\Gtwo,3}\subset \A_{\grE_6,1}$ 
  by $\alpha$-induction
  \cite{XuUnpublished}, see also
\cite[Appendix]{CaMoSn2011}. 
The subfactor was also constructed by Izumi in
 \cite{Iz2000}.
Note that $\Rep(\A_{\SU(3),1})\cong\rev{\Rep(  \A_{\grE_6,1})}$, thus
  by Proposition \ref{prop:GaloisDouble} (3)
  the net $\A_{\Gtwo,3}\otimes \A_{\SU(3),1}$ realizes its quantum double.
  A similar observation was made by
  Ostrik \cite[Remark A.4.3]{CaMoSn2011}. 
  The complex conjutage should be realized by
  $\tilde \A_{\Gtwo,3}\otimes \A_{\grE_6,1}$, but we do not know how to realize 
  the net $\tilde \A_{\Gtwo,3}$.
\end{itemize}
\end{rmk}

\subsection{Modular invariants}
All our examples in Corollary \ref{cor:Doubles} are $\ZZ_2$--simple current extension.
We remember that for $\A\subset \cB$ an extension, $N=\A(I)\subset \cB(I)=M$ and 
$\bim N\cC N=\Rep^I(\A)$, the matrix $Z=(Z_\mu,\nu)_{\mu,\nu\in\Irr(\bim N\cC N)}$
 with $Z_{\mu,\nu}=  \dim\Hom(\alpha^+_\mu,\alpha^-_\nu)$ is a modular invariant \cite{BcEvKa1999}, \ie commutes
with the $S$ and $T$ associated with $\bim N\cC N$.
The modular invariant of a commutative $\ZZ_2$--simple current extension $\theta=[\rho_0]\oplus[\rho_g]$ is given by
(\cf (3.59) in \cite{FuRuSc2004} for the general formula)
$$
  Z_{i,j}=\frac12\left(1+\frac{\omega_{gi}}{\omega_i}\right)\left(\delta_{i,j}+\delta_{gi,j}\right)\,,
$$
where $gi$ is the action of $g$ on $i$, \ie $[\rho_{gi}]=[\rho_g]\times[\rho_i]$. We conveniently write the modular invariant in character form as:
$$
  Z=\sum_{\mu,\nu} Z_{\mu,\nu}\chi_\mu\bar\chi_\nu\,.
$$
We include the modular invariants, from which one can derive the fusion rules of the representation category. We note, although it is not necessary and follows from the above ``abstract non-sense'', one can directly check
that the, for example the representation category of the net $\A_{10}\otimes \A_{\Spin(11),1} \myrtimes_{[\rho_{10,2}]} \ZZ_2$
has the fusion rules of the $E_6$ double as in \cite{Iz2001II,HoRoWa2008}. Some of this calculation is contained in \cite{BcEvKa2001}.

\begin{example}[$A_{k+1}$-case] 
For the inclusion $\A_k\otimes \tilde \A_k\subset \A_{N\subset N}= (\A_k\otimes \tilde \A_k)\rtimes \ZZ_2$
has
the modular invariant is given by: 
$$
Z_{\rho_{i_1,j_1},\rho_{i_2,j_2}}
=\frac12\left(1+(-1)^{i_1-j_1}\right)\left(\delta_{i_1,i_2}\delta_{j_1,j_2}
  +\delta_{i_1,k-i_2}\delta_{j_1,k-j_2}\right)
$$
and thus
$$
  Z=\frac12\sum_{\substack{i,j=0 \\ i+j=\mathrm{even}}}^k |\chi_{\rho_{i,j}}+\chi_{\rho_{k-i,k-j}}|^2\,.
$$
\end{example}
\begin{example}[$D_{2n}$-case] 
Let $k=4n-4$. Let $\A_k$, then there is a simple current extension  
$\cB_k=\A_{k}\rtimes_{\rho_{k}} \ZZ_2$ of $\A_k$ corresponding to the Dynkin diagram $D_{2n}$ in Proposition \ref{prop:ADE}
  with modular invariant:
$$
  Z_{D_{2n}}=\frac12\sum_{\ell=0}^{\frac k2} |\chi_{2\ell}+\chi_{k-2l}|^2 \,.
$$
The same is true for $\tilde \cB_k=\tilde \A_{k}\rtimes_{\rho_{k}} \ZZ_2$.
The net $\A_{N\subset M}$ for $D_{2n}$ is just $\cB_k\otimes \tilde\cB_k$, which is an $\ZZ_2$ extension of 
$$\A_k\otimes \tilde \cB_k \subset \cB_k\otimes \tilde \cB_k
\supset \cB_k\otimes \tilde \A_k\,.$$ 
So the modular invariant for the $\ZZ_2$-simple current extension is
$Z_{D_{2n}}\otimes I_{n+1}$, where $I_m$ is the $m\times m$ identity matrix.
\end{example}

\begin{example}[$E_{6}$-cases] 
Then modular invariant for 
$\A_{\SU(2),{10}}\otimes\A_{\Spin(11),1}\subset \A_{N\subset M}$ for $E_6$
is given by:
 $$Z=X+Y +2  |\chi_{5,1}|^2\,,$$
 with
\begin{align*}
  X&= 
  |\chi_{0,0}+\chi_{10,2}|^2 +
  |\chi_{0,2}+\chi_{10,0}|^2 + 
  |\chi_{2,0}+\chi_{8,2}|^2 +
  |\chi_{0,2}+\chi_{8,0}|^2 +
  |\chi_{4,0}+\chi_{6,2}|^2 +
  |\chi_{4,2}+\chi_{6,0}|^2 
\\
  Y&=
  |\chi_{1,1}+\chi_{9,1}|^2 +
  |\chi_{3,1}+\chi_{7,1}|^2 \,.
\end{align*}
One can read of the number of irreducible sectors:
  $|\bim N\Delta N|=33$,
  $|\bim N\Delta M|=|\bim M{\Delta^\pm} M|=18$, 
  $|\bim M\Delta M|=36$ and 
  $|\bim M{\Delta^0} M|=10$.
The category $\bim N\cC N$ has $A_{11}\times A_{3}$ fusion rules, see Figure \ref{fig:A11A3}
and the $\ZZ_2$--simple current extension is an ``orbifold'' 
giving the fusion rules of $\bim[\pm] M\cC M$, Figure \ref{fig:E6DualFusionGraph}.
\begin{figure}
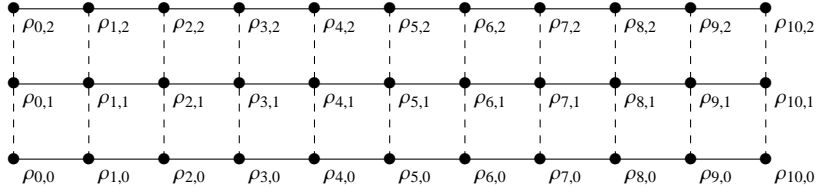

$$
  \tikzmath[1]{
    \foreach \k in {0,...,10}{
      \foreach \l in {0,1,2}{
        \node at (\k,\l) {$\bullet$};
        \node at (\k,\l) [below right] {$\scriptstyle{\rho_{\k,\l}}$};
      }
    }
    \foreach \k in {0,...,10}{
      \draw[dashed] (\k,0)--(\k,2);
    }
    \foreach \l in {0,1,2}{
      \draw (0,\l)--(10,\l);
    }
  }
$$
\caption{Fusion rules of $\SU(2)_{10}\times \Spin(11)_1$}
\label{fig:A11A3}
\end{figure}
\begin{figure}
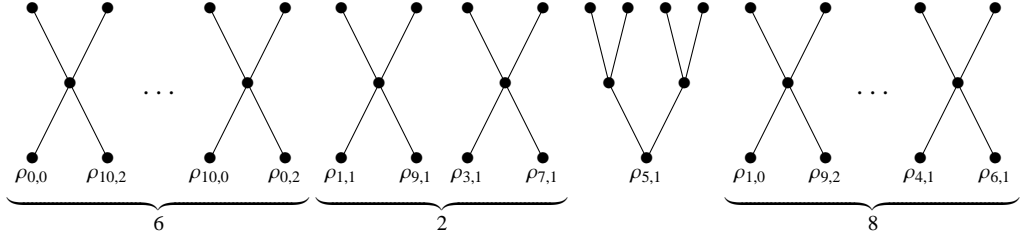

\label{fig:E6SimplecurrentPrincipalGraph}
$$
  \underbrace{
  \tikzmath[1]{
    \node at (0,0) {$\bullet$};
    \node at (1,0) {$\bullet$};
    \node at (.5,1) {$\bullet$};
    \node at (0,2) {$\bullet$};
    \node at (1,2) {$\bullet$};
    \node at (0,0) [below] {$\scriptstyle{\rho_{0,0}}$};
    \node at (1,0) [below] {$\scriptstyle{\rho_{10,2}}$};
    \draw (0,0)--(.5,1);
    \draw (1,0)--(.5,1);
    \draw (0,2)--(.5,1);
    \draw (1,2)--(.5,1);
  }
  \cdots
  \tikzmath[1]{
    \node at (0,0) {$\bullet$};
    \node at (1,0) {$\bullet$};
    \node at (.5,1) {$\bullet$};
    \node at (0,2) {$\bullet$};
    \node at (1,2) {$\bullet$};
    \node at (0,0) [below] {$\scriptstyle{\rho_{10,0}}$};
    \node at (1,0) [below] {$\scriptstyle{\rho_{0,2}}$};
    \draw (0,0)--(.5,1);
    \draw (1,0)--(.5,1);
    \draw (0,2)--(.5,1);
    \draw (1,2)--(.5,1);
  }
  }_6
  \underbrace{
  \tikzmath[1]{
    \node at (0,0) {$\bullet$};
    \node at (1,0) {$\bullet$};
    \node at (.5,1) {$\bullet$};
    \node at (0,2) {$\bullet$};
    \node at (1,2) {$\bullet$};
    \node at (0,0) [below] {$\scriptstyle{\rho_{1,1}}$};
    \node at (1,0) [below] {$\scriptstyle{\rho_{9,1}}$};
    \draw (0,0)--(.5,1);
    \draw (1,0)--(.5,1);
    \draw (0,2)--(.5,1);
    \draw (1,2)--(.5,1);
  }
  \tikzmath[1]{
    \node at (0,0) {$\bullet$};
    \node at (1,0) {$\bullet$};
    \node at (.5,1) {$\bullet$};
    \node at (0,2) {$\bullet$};
    \node at (1,2) {$\bullet$};
    \node at (0,0) [below] {$\scriptstyle{\rho_{3,1}}$};
    \node at (1,0) [below] {$\scriptstyle{\rho_{7,1}}$};
    \draw (0,0)--(.5,1);
    \draw (1,0)--(.5,1);
    \draw (0,2)--(.5,1);
    \draw (1,2)--(.5,1);
  }
  }_{2}
  \tikzmath[1]{
    \node at (.5,0) {$\bullet$};
    \node at (0,1) {$\bullet$};
    \node at (1,1) {$\bullet$};
    \node at (-.25,2) {$\bullet$};
    \node at (.25,2) {$\bullet$};
    \node at (.75,2) {$\bullet$};
    \node at (1.25,2) {$\bullet$};
    \node at (.5,0) [below] {$\scriptstyle{\rho_{5,1}}$};
    \draw (0.5,0)--(0,1);
    \draw (.5,0)--(1,1);
    \draw (0,1)--(-.25,2);
    \draw (0,1)--(.25,2);
    \draw (1,1)--(.75,2);
    \draw (1,1)--(1.25,2);
  }
  \underbrace{
  \tikzmath[1]{
    \node at (0,0) {$\bullet$};
    \node at (1,0) {$\bullet$};
    \node at (.5,1) {$\bullet$};
    \node at (0,2) {$\bullet$};
    \node at (1,2) {$\bullet$};
    \node at (0,0) [below] {$\scriptstyle{\rho_{1,0}}$};
    \node at (1,0) [below] {$\scriptstyle{\rho_{9,2}}$};
    \draw (0,0)--(.5,1);
    \draw (1,0)--(.5,1);
    \draw (0,2)--(.5,1);
    \draw (1,2)--(.5,1);
  }
  \cdots
  \tikzmath[1]{
    \node at (0,0) {$\bullet$};
    \node at (1,0) {$\bullet$};
    \node at (.5,1) {$\bullet$};
    \node at (0,2) {$\bullet$};
    \node at (1,2) {$\bullet$};
    \node at (0,0) [below] {$\scriptstyle{\rho_{4,1}}$};
    \node at (1,0) [below] {$\scriptstyle{\rho_{6,1}}$};
    \draw (0,0)--(.5,1);
    \draw (1,0)--(.5,1);
    \draw (0,2)--(.5,1);
    \draw (1,2)--(.5,1);
  }
  }_{8}
$$
\caption{(Dual) principal graph for $\A_{\SU(2),{10}}\otimes\A_{\Spin(11),1}\subset \A_{N\subset M}$} 
\end{figure}
\begin{figure}
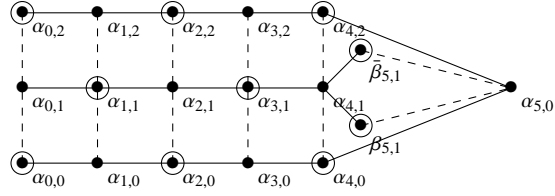

$$
  \tikzmath[1]{
    \foreach \k in {0,...,4}{
      \foreach \l in {0,1,2}{
        \pgfmathparse{Mod(\k+\l,2)==0?1:0}
         \ifnum\pgfmathresult>0
            \draw (\k,\l) circle (.15);
            \node at (\k,\l) {$\bullet$};
           \else
            \node at (\k,\l) {$\bullet$};
           \fi
           \node at (\k,\l) [below right] {$\scriptstyle{\alpha_{\k,\l}}$};
      }
    }
    \foreach \k in {0,...,4}{
      \draw[dashed] (\k,0)--(\k,2);
    }
    \foreach \l in {0,1,2}{
      \draw (0,\l)--(4,\l);
    }
    \node at (4.5,.5) {$\bullet$};
    \node at (4.5,.5) [below right] {$\scriptstyle{\beta_{5,1}}$};
    \node at (4.5,1.5) {$\bullet$};
    \node at (4.5,1.5) [below right] {$\scriptstyle{\bar\beta_{5,1}}$};
    \draw (4.5,.5)--(4,1)--(4.5,1.5); 
    \node at (6.5,1) {$\bullet$};
    \node at (6.5,1) [below right] {$\scriptstyle{\alpha_{5,0}}$};
    \draw[dashed] (4.5,.5)--(6.5,1)--(4.5,1.5); 
    \draw (4,0)--(6.5,1)--(4,2); 
    \draw (4.5,.5) circle (.15);
    \draw (4.5,1.5) circle (.15);
  }
$$
\caption{The fusion graph of $\bim[+]M\cC M$ for
 $\A_{\SU(2),{10}}\otimes\A_{\Spin(11),1}\subset \A_{N\subset M}$ for $E_6$}
\label{fig:E6DualFusionGraph}
\end{figure}
\end{example}

\begin{example}[$E_{8}$-cases] 
Note the net $\cB_{N\subset M}$ for the $E_8$ subfactor can be realized as $\A_{D_{16}}\otimes \A_{\Ffour,1}$, 
where we can replace $\A_{\Ffour1}$ by the $\tilde\cB\supset\tilde \A_{28}$ the 
extension from Proposition \ref{prop:MirrorExtension}
of $\cB=\A_{\Gtwo,1}\supset \tilde \A_{28}$.
The modular invariant of the inclusion 
$\A_{28}\otimes \A_{\Ffour,1} \subset \A_{D_{16}}\otimes \A_{\Ffour,1}$ is
$Z_{D_{12}}\otimes I_{2}$. 
\end{example}

\section{Categorical Picture and Vertex Operator Algebras}
\label{sec:VOA}
Local irreducible extensions $\cB\supset \A$ of completely rational nets are characterized by 
commutative Q-systems $\Theta\in  \Rep(\A)$ 
\cite{LoRe1995}
and the representation theory is given by the ambichiral sectors $\bim[0]M\cC M$.
The Q-system is a commutative (Frobenius) algebra in the braided tensor category $\Rep(\A)$.
Because $\Theta$ is commutative, the right-modules $\Mod(\Theta)=\cC_\Theta$, see \cite{KiOs2002}, form itself 
a tensor category. This category is equivalent with $\bim[+] M\cC M$. Interchanging the braiding, there is another 
tensor product under which $\Mod(\Theta)$ is equivalent with $\bim[-] M\cC M$.
The ambichiral sectors are braided equivalent  
$\bim[0]M\cC M$ with the category of local or dyslexic modules $\Mod_0(\Theta)$, see 
\cite{BiKaLo2014}.

The same categorical structure arises for extensions of vertex operator algebras.
\cite{KiOs2002, HuKiLe2014}. It follows:
\begin{prop}
\label{prop:VOAext}
 Let $\A$ be a completely rational conformal net and 
$V$ a vertex operator algebra, such that the category $\cC_V$ has a natural 
vertex tensor category structure (\cf \cite{HuKiLe2014}) and is braided equivalent to $\Rep(\A)$. 
Then for every local irreducible extension $\cB\supset\A$ there exists a vertex operator 
algebra $V_\cB\supset V$, whose category of modules is braided equivalent to $\Rep(\cB)$.
\end{prop}

Using this proposition we can transport our result to vertex operator algebras. 
By 
\cite[Proposition 8.2.6]{FrKe1993} ribbon categories with $\SU(2)_k$ 
are determined by its twists which are given by the exponential of the conformal weights using \cite{GuLo1996}. The fusion rules calculated by \cite{Wa} coincide
with the one of the corresponding affine Kac--Moody VOA. Thus we can conclude that the 
modular tensor categories are equivalent.

For a VOA corresponding to the net $\A_k$, \ie a VOA which has the opposite braiding of $\SU(2)_k$, 
we could in principle apply Proposition \ref{prop:VOAext}, but we do not know that the 
categories for the Virasoro minimal models are equivalent for VOAs and conformal nets.

But we can argue as follows. Let $V_k=V_{\SU(2)_k}$ be the vertex operator algebra of affine Kac-Moody algebra $\hat{\mathfrak{sl}}_2$
at level $k$. As in Proposition \ref{prop:Akmirror} we get an inclusion into $V_{E_8}^{\otimes k}$, where 
$V_{E_8}$ is the vertex operator algebra associated with the even lattice $E_8$, which coincides by the Kac--Frenkel construction with the 
affine Kac-Moody algebra of the Lie algebra $E_8$ at level 1. Let $\tilde V_k$ be the coset of the inclusion $V_k\subset V_{E_8}^{\otimes k}$. Then 
$V_{E_8}^{\otimes k}$ decomposes as 
$$
  \bigoplus Z_{kl} M_k\otimes \tilde M_l  \,,
$$
where $M_k$ are modules of $V_{k}$ and $\tilde M_l$ of the coset net $\tilde V_l$. It is $Z_{k0}=\delta_{k,0}$ and $Z_{0l}=\delta_{l,0}$.
We call such an inclusion of $V_k\subset V_{E_8}^{\otimes k}$ normal.
By the same argument as in Proposition \ref{prop:VOAext}
 the analogue of Proposition \ref{prop:Mirror} holds using the same proof and $\tilde V_k$ has as representation category $\SU(2)_k$
with the opposite braiding.

Then Corollary \ref{cor:Doubles} together with Proposition \ref{prop:VOAext} gives:
\begin{prop} 
There is a unitary rational VOA $\tilde V_k$ which has the opposite braiding of $\SU(2)_k$.

For every subfactor $[M:N]<4$ there is a unitary rational VOA $V_{N\subset M}$,   
whose category of modules is equivalent to the quantum double $D(N\subset M)$ of the subfactor $N\subset M$, \ie 
the Drinfeld center of the fusion category of the even part of $N\subset M$. 
\end{prop}
\begin{rmk} For the construction of $\tilde V_k$ and $V_{N\subset M}$ we could also use directly the correspondence 
between conformal nets and vertex operator algebras in \cite{CaKaLoWi2015}. We still have to use the categorical arguments
to show that the corresponding representations categories are equivalent. It would be nice to have a result that states 
that the representation categories of $V$ and $\A_V$ are the same.
\end{rmk}
\begin{example} 
Let $V$ be the vertex operator algebra obtained by $\ZZ_2$--simple current extension 
$\hat{\mathfrak{sl}}_{2,10}\otimes \hat{\mathfrak{so}}_{11,1}$.
Then the category modules of $V$ 
is equivalent to $Z(\frac12E_6)$, the quantum double of the $E_6$ subfactor.
\end{example}

\section{Conclusions and Outlook}
We gave some structural results of completely rational conformal nets whose representation category is a quantum double (Drinfeld center of 
a unitary fusion category). 
We showed that the quantum doubles of subfactors with index less than 4, or equivalently the Drinfeld centers of 
their even part fusion categories, are realized as representation theories in chiral conformal field theory,
either as a conformal net of von Neumann algebras or as VOAs.
 The most interesting is the realization of the quantum double of $E_6$ (or $\bar E_6$) 
as a $\ZZ_2$-simple current extension of $\SU(2)_{10}\times \Spin(11)_1$. In particular,
\cite{HoRoWa2008} it was shown that the quantum double of $E_6$ is universal for topological quantum computing.
On the other hand, it was proposed in the same article that it might be exotic. Our construction shows that it is indeed not exotic.
This example was the main motivation of the article, because no direct realization in conformal field theory or
quantum groups is contained in the literature. Further, the even part of $E_6$ is the smallest non-trivial fusion category
\cite{Os2013} in the sense that it is not braided or coming from groups. Drinfeld centers of braided fusion categories and groups are 
easy. Despite the fact that the even part of $E_6$ has no braiding, the realization as a CFT is 
still very easy.

We conjecture that the double of $E_6$ is also related to Chern--Simons theory  with non-simply connected 
gauge group $(\SU(2)\times \Spin(11)) /\ZZ_2$. It is also related to the 
$\SU(2)_{10}\times \Spin(11)_1$ quantum group as a kind of quantum subgroup. Indeed the $\ZZ_2$-simple current extension
correspond to a quantum subgroup in the sense of Ocneanu \cite{Oc2001}.

It would be interesting to find realizations of the doubles of exotic subfactors, like the Haagerup subfactor using similar 
methods like here.

\def\cprime{$'$}


\address
\end{document}